\documentclass[a4paper,11pt]{article}
\usepackage{xspace}
\usepackage[left=2cm,right=2cm,top=2cm,bottom=2cm]{geometry}
\usepackage{amsthm,amsmath,amssymb}
\usepackage{enumitem}
\usepackage{tikz}
\usepackage{tkz-graph}
\usepackage{tikz}
\usepackage{tikz-3dplot}
\usetikzlibrary{calc,intersections,through,backgrounds}
\usetikzlibrary[arrows,decorations.pathmorphing,backgrounds,positioning,fit,petri]
\usepackage{tkz-graph}
\usepackage{calc}
\usetikzlibrary{arrows}

\makeatletter
\let\old@setaddresses\@setaddresses
\def\@setaddresses{\bgroup\parindent 0pt\let\scshape\relax\old@setaddresses\egroup}
\makeatother
\usepackage{algorithmic}
\usepackage{algorithm}
\def\Re{\mathbb R}
\newcommand{\VC}{\ensuremath{\mathsf{VC}}\xspace}

\newcommand{\BARANY}{B{\' a}r{\' a}ny\xspace}
\newcommand{\CARA}{Carath{\' e}odory\xspace}
\newcommand{\Matousek}{Matou{\v s}ek\xspace}
\theoremstyle{plain}

\newcommand{\cardin}[1]{\lvert {#1} \rvert}

\newcommand{\F}{\mathcal{F}}

\theoremstyle{definition}
\newtheorem{definition}{Definition}

\newtheorem{lemma}{Lemma}
\newtheorem{theorem}{Theorem}
\newtheorem{remark}{Remark}
\begin{document}

%\settowidth\leftmargini{(9$''$)\hskip\labelsep}
\bibliographystyle{plain}
%\thispagestyle{empty}
%\vspace*{\fill}
%\begin{center}
%\LARGE
%\textbf{

\title{A short proof of the first selecttion lemma and weak $\frac{1}{r}$-nets for moving points.\footnote{Work was partially supported by Grant 1136/12 from the Israel Science Foundation}}
\author{Alexandre Rok, Shakhar Smorodinsky\thanks{Work by this author was partially supported by Swiss National Science Foundation Grants 200020144531 and 200021-137574.}}

\maketitle
%\par\vspace{5ex}
	
%\par\vspace{5ex}
%\Large
\begin{abstract}
(i) We provide a short and simple proof of the first selection lemma.
(ii) We also prove a selection lemma of a new type in $\Re^d$. For example, when $d=2$ assuming $n$ is large enough we prove that for any set $P$ of $n$ points in general position there are $\Omega(n^4)$ pairs of segments spanned by $P$ all of which intersect in some fixed triangle spanned by $P$. (iii) Finally, we extend the weak $\frac{1}{r}$-net theorem to a kinetic setting where the underlying set of points is moving polynomially with bounded description complexity. We establish that one can find a kinetic analog $N$ of a weak $\frac{1}{r}$-net of cardinality $O(r^{\frac{d(d+1)}{2}}\log^{d}r)$ whose points are moving with coordinates that are rational functions with bounded description complexity. Moreover, each member of $N$ has one polynomial coordinate.

 %there exists a spanned triangle $\Delta$ such that there are $\Omega(n^4)$ Tverberg partions of $4$-tuples from $P$ each determining a Tverberg point, i.e. the crossing of two spanned segments or a point of $P$, in $\Delta$.

 %We also present an alternative approach through the notion of centerpoint leading to a slightly worse bound. Finally, we provide a better bound with the notion of Radon points in  $\mathbb{R}^3$ of $O(r^{5}\log^{5}r)$.
\end{abstract}

\section{Introduction and Preliminaries}
\label{intro}

This paper deals with the so-called first selection lemma and weak $\frac{1}{r}$-nets for convex sets. These are two central notions in discrete geometry.
In the first part of the paper, we provide a short and almost elementary proof of the first selection lemma. As a corollary of our proof technique, we also obtain another type of a selection lemma that we believe is of independent interest.
In the second part of this paper we initiate the study of kinetic weak $\frac{1}{r}$-nets. We extend the classical weak $\frac{1}{r}$-net theorem to a kinetic setting. Our main motivation is the recent result of De Carufel et al. \cite{Ca15} on kinetic hypergraphs.

Before presenting our results we need a few definitions and well known facts:
A pair $(X, \mathcal{S})$, where $\mathcal{S}\subset P(X)$, is called a \emph{set system} or a \emph{hypergraph}. A subset $A\subset X$ is called {\em shattered} if $\mathcal{S}|_{A}= 2^{A}$. The maximum size of a shattered subset from $X$ with respect to $\mathcal{S}$ is called the \emph{\VC-dimension} of $(X, \mathcal{S})$.
The concept of \VC-dimension has its roots in statistics. It first appeared in the paper of Vapnik and Chervonenkis in \cite{VaCh71}. Nowadays, this notion plays a key role in learning theory and discrete geometry.
Given a set system $(X, \mathcal{S})$,   we say that $Y\subset X$ is a \emph{strong $\frac{1}{r}$-net} if for each $S\in \mathcal {S}$ with $|S|> |X|/r$ we have $S\cap Y \neq \emptyset$.  Based on the concept of \VC-dimension, Haussler and Welzl provided a link with strong nets by proving that any set system with \VC-dimension $d$ has a strong $\frac{1}{r}$-net of size $O(dr\log r)$ \cite{HaWe87}. We say that a set of points $P\subset\Re^d$ is in {\em general position} if no $d+1$ of them are contained in a hyperplane. Given two points $x=(x_1,\ldots, x_d)$ and $y=(y_1,\ldots, y_d)$ we write $x<_{\textnormal{lex}}y$ if there is an $1\leq i\leq d$ such that for all indices $j<i$ we have $x_j= y_j$ and $x_i < y_i$. It is not hard to see that $<_{\textnormal{lex}}$ determines a total order on $\Re^d$, and that any compact set $S \subset \Re^d$ admits a unique minimum with respect to $<_{\textnormal{lex}}$, which is called \emph{the lexicographic minimum} of $S$. We denote the convex hull of a set $X \subset \Re^d$ by conv$(X)$, and affine hull of a finite $X $ by aff$(X)$. For a set of points $P \subset \Re^d$ in general position, we refer to the convex hull of a nonempty subset $S \subset P$ with $1 \leq \cardin{S} \leq d+1$ as the {\em ($\cardin{S}-1)$-dimensional simplex spanned by $S$}. A simplex is \emph{spanned by} $P$ if it is spanned by some subset of $P$.

%The points of $X$ are said to be \emph{affinely independent} if for each $x\in X$ we have $x\not\in\textnormal{aff}(X\setminus \{x\})$, otherwise they are said to be \emph{affinely dependent}. A set $L\subset \mathbb{R}^d$ is called a \emph{linear space} if it is closed under addition and under multiplication by real numbers. A set $A\subset \mathbb{R}^d$ is called an \emph{affine space} if $A=L+v$, where $L$ is some linear space and $v\in \mathbb{R}^d$.
\subsection{The first selection lemma}
Using simple arguments we provide an alternative short proof to the following classical theorem in discrete geometry known as the {\em first selection Lemma}:
\begin{theorem}[\bf First Selection Lemma]\label{first-selection}
There exists a constant $c=c(d)>0$ such that the following holds:
Let $P\subset \Re^d$ be a set of $n$ points in general position. Then there exists a point contained in at least $c {n \choose {d+1}}$ $d$-dimensional simplices spanned by $P$.
\end{theorem}

 The planar version of Theorem~\ref{first-selection} with the (best possible) constant $c(2)= \frac{2}{9}$ was proved by Boros and F{\" u}redi in \cite{BF84}. A generalization to an arbitrary dimension was proved by \BARANY in \cite{Bar82} with $c(d) \geq \frac{1}{(d+1)^d}$.
\BARANY used a clever combination of two deep theorems: Tverberg's Theorem \cite{Tve66} and Colorful \CARA's Theorem \cite{Bar82}.
Several other proofs are known (cf. \cite{MATOUSEK}). All proofs combine Tverberg's Theorem with another tool.
A surprising deep generalization to a topological setting was obtained by Gromov \cite{GROMOV} with an even better constant $c(d) \geq \frac{2d}{(d+1)(d+1)!}$. See also the recent exposition of Gromov's method by \Matousek and Wagner \cite{MatWag}.
The best known estimate on $c(3)$ was obtained by Basit et al. in \cite{BMRR10} using clever elementary arguments. Determining the dependency of $c(d)$ on the dimension $d$ remains a wide open problem.
The best known upper bound $c(d) \leq \frac{(d+1)!}{(d+1)^{d+1}}= e^{-\Theta(d)}$ was obtained by Bukh, \Matousek, and Nivasch \cite{BMN10}. In this paper, we provide a simple proof using Tverberg's Theorem in a "very weak" sense described in Section~\ref{sec:short-pf}.

\subsection{A simplex containing many points lemma}
 A point $p$ is said to be \emph{a Tverberg point} for the set $P$ with $|P|=(r-1)(d+1)+1$ if there exists a partition of $P$ into $P_1, \ldots, P_{r}$ with $p\in \bigcap_{i=1}^r\textnormal{conv}(P_i)$. This intersection may contain many points, but its lexicographic minimum is unique. Hence, we say that $p$ is \emph{determined} by the partition above if it is the lexicographic minimum of $\bigcap_{i=1}^r\textnormal{conv}(P_i)$. Using our proof technique (for the first selection lemma) we provide yet another type of a "Selection Lemma":
\begin{theorem}[\bf Simplex Containing Many Points Lemma]\label{rich-simplex}

Let $P\subset \mathbb{R}^d$ be a set of $n\geq d^2+d+1$ points in general position. Then there exists a $d$-simplex $\Delta$ spanned by $P$ and a set of at least $\Omega \left ({n \choose d^2} \right)$ $d^2$-tuples from $P$ each having a Tverberg partition into $d$ sets determining a Tverberg point in $\Delta$.
\end{theorem}

%if one counts a Tverberg point once for each $d^2$-tuple from which it arises
For example, for $d=2$ Tverberg points above are simply Radon points so these are either points of $P$ or crossings of two spanned segments. Hence, Theorem~\ref{rich-simplex} implies that for every set $P$ of $n\geq d^2+d+1$ points in general position in the plane there exists a triangle $\Delta$ spanned by $P$ that either contains $\Omega(n^4)$ points of $P$ counted with multiplicities (each point is counted once for each Radon partition of a $4$-tuple it arises from) or  $\Omega(n^4)$ pairs of segments spanned by $P$ intersect in $\Delta$. Since a point of $P$ can be the Radon point of at most $O(n^3)$ Radon partitions of $4$-tuples, we see that in the former case $\Delta$ contains $\Omega(n)$ distinct points of $P$. However, by the well known crossing lemma \cite{ACNS82,Leighton83} every convex set containing $m$ (for $m$ large enough) points also contains $\Omega(m^4)$ of its crossing points. Hence, for $n$  large enough Theorem~\ref{rich-simplex} implies  that $\Omega(n^4)$ pairs of segments spanned by $P$ intersect in $\Delta$.

\subsection{Kinetic weak $\frac{1}{r}$-nets}
Next, we study the notion of weak $\frac{1}{r}$-net in a kinetic setting. Let us first recall the concept of weak $\frac{1}{r}$-net in the static case.

\begin{definition}[\bf Weak $\frac{1}{r}$-net]
Let $P\subset \mathbb{R}^d$  be a finite set of points and $r\geq 1$. A set $N\subset \Re^d$ is said to be a weak $\frac{1}{r}$-net for $P$ if every convex set containing  $>\frac{1}{r} \cardin{P}$ points of $P$ also contains a point of $N$.
\end{definition}
The following theorem is one of the major milestones in modern discrete geometry:
\begin{theorem}[\bf Weak $\frac{1}{r}$-net Theorem \cite{ABFK,CEGGSW,MW}]\label{thm:weak-net}
Let $d,r\geq 1$ be integers. Then there exists a least integer $f(r,d)$ such that for every finite set $P\subset \Re^d$ there is a weak $\frac{1}{r}$-net of size at most $f(r,d)$.
\end{theorem}

The existence of $f(r,d)$ was first proved by Alon et al. \cite{ABFK} with the bounds $f(2,r)=O(r^2)$ and $f(d,r)=O(r^{(d+1)(1- \frac{1}{s_d})})$ for $d \geq 3$, where $s_d$ tends to $0$ exponentially fast. Later, better bounds on $f(r,d)$ for $d\geq 3$ were obtained by Chazelle at al. in \cite{CEGGSW}, who showed that $f(d,r)=O(r^d\log^{b_d} r)$, where $b_d$ is roughly $2^{d-1}(d-1)!$. The current best known upper bound for $d\geq 3$ due to Matou{\v s}ek and Wagner \cite{MW} is
$f(r,d)= O(r^d \log^{c(d)} r)$, where $c(d) = O(d^2 \log d)$,  and
$f(r,2) = O(r^2)$ \cite{ABFK}. The best known lower
bound was provided by Bukh, Matou{\v s}ek, and
Nivasch \cite{BMN11}, who showed that $f(d,r) = \Omega(r \log^{d-1} r)$ for $d \geq 2$. Recently, some interesting connections were found between strong and weak nets. In particular, Mustafa and Ray \cite{MuRa10} showed how one can construct weak $\frac{1}{r}$-nets from strong $\frac{1}{r}$-nets. They obtained a bound of $O(r^3\log^3r)$ in $\mathbb{R}^2$, $O(r^5\log^5r)$ in $\mathbb{R}^3$, and provided a bound of $O(r^{d^2}\log^{d^2}r)$ for $d\geq 4$ on the size of weak $\frac{1}{r}$-nets.

{\noindent \bf A kinetic framework:}
The problem of finding strong $\frac{1}{r}$-nets has been recently considered in a kinetic setting  by De Carufel et al. \cite{Ca15}. Their work and extensive research in the static case motivates us to consider the problem of weak $\frac{1}{r}$-net in a kinetic setting.
Let us define this setting:
A \emph{moving point} is a  function from $\mathbb{R}_{+}$ to $ \mathbb{R}^d\cup \{\emptyset\}$ for some $d\geq 1$. In this paper, we are interested in the case where this function is polynomial or rational, i.e., each coordinate is a polynomial or a rational function. If one of the coordinates is not defined for some $t$, then the moving point is not defined. For simplicity, we often use the term point for a moving point if there is no confusion. In what follows, the dimension $d$ is assumed to be fixed. For a set $P$ of moving points and a ``time" $t\in \mathbb{R}_{+}$, we denote by $P(t)$ the set $\{p(t) | p \in P\}$. We say that a set $P$ of moving points in $\mathbb{R}^d$ has bounded description complexity $\beta$ if for each point $p(t)=(p_{1}(t),\ldots, p_{d}(t))$, each $p_i(t)$ is a rational function with both numerator and denominator having degree at most $\beta$. We say that the function $h$ with domain $\mathbb{R}_{+}$ is a \emph{moving affine space} if for some integer $k$ and any $t\geq 0$ $h(t)$ is an affine space of dimension $k$ or the emptyset. In case $h(t)$ is not always equal to the emptyset, we also say that such a $h$ has \emph{dimension} $k$. If $k=1$ (respectively $k=d-1$) we refer to the corresponding moving affine space as a \emph{moving line} (respectively, a \emph{moving hyperplane}). For simplicity, we often write a moving space. Similarly to moving points, if a moving space $h$ is given by $x_1=p_{1},\ldots, x_{k}=p_{k}$, where $p_i: \mathbb{R}_{+} \to \mathbb{R}\cup \{\emptyset\}$ are moving points and $p_i(t)$ is not defined for some $t\geq 0$ then $h(t)$ is not defined. For a set $P=\{p_1,\ldots, p_n\}$ of moving points in $\mathbb{R}^d$ and a vector space $V\subset \mathbb{R}^d$, we say that $P'=\{p'_1,\ldots, p'_n\}$ is a \emph{projection of $P$ onto $V$} if $p'_i(t)=\textnormal{proj}_{V}(p_i(t))$ for all $t\geq 0$.
\begin{definition}[\bf Kinetic Weak $\frac{1}{r}$-net:]
Given a set $P$ of points moving in $\mathbb{R}^d$, we say that a set of moving points $N$ is a kinetic weak $\frac{1}{r}$-net for $P$ if for any $t\in \mathbb{R}_{+}$ and any convex set $C$ with $C\cap P(t)> n/r$ we have $C\cap N(t) \neq \emptyset$.
\end{definition}
% We also recall the assumption that for each $t\geq0$ no set of $d+2$ points of $P$ is contained in a hyperplane.

We sometimes abuse the notation and write net or weak net instead of kinetic weak net. In order to establish our result regarding kinetic nets, we need the following natural general position assumption on the set $P$ of moving points: We assume that for any $t\geq 0$ the affine hull of any $d$-tuple of points in $P(t)$ is a hyperplane, but no $d+2$ points of $P(t)$ are contained in a hyperplane. The latter can easily be relaxed to no $c(d)\geq d+2$ points in a hyperplane. Under these natural assumptions, we prove the following theorem that could be viewed as a generalization of Theorem~\ref{thm:weak-net}:

\begin{theorem}[\bf Kinetic Weak $\frac{1}{r}$-net Theorem:]\label{thm:kinetic-weak-nets}
For every triple of integers $r, d$ and $\beta$ there exist $c(r,d,\beta)$ and $g(d,\beta)$ such that for every finite set $P$ of moving points in $\Re^d$ with description complexity $\beta$ there is a kinetic weak $\frac{1}{r}$-net of cardinality at most $c(r,d,\beta)$ and description complexity $g(d,\beta)$. Moreover,  for fixed $d$ and $\beta$ and $r\geq 2$ we have $c(r,d, \beta)=O(r^{\frac{(d+1)d}{2}} \log^d r)$.
\end{theorem}

%Theorem~\ref{rich-simplex} combined with a recent result of De Carufel et al. \cite{DECARUFEL-ETAL}(see Theorem~\ref{kinetic-VC} below) enables us to provide a simple proof of Theorem~\ref{thm:kinetic-weak-nets} with the bounds $f(r,d,\beta)=O(r^{d^2}\log r)$ where the constant in the big-`Oh' notation depends on $\beta$.
Moreover, in the case where the points of $P$ move polynomially, the moving points of the kinetic weak $\frac{1}{r}$-net have one polynomial coordinate.
This is an important advantage of our construction as many naturally defined moving points, obtained by intersecting moving affine spaces, have no polynomial coordinates.

\section{A Short Proof of the First Selection Lemma}\label{sec:short-pf}
Let us first recall Tverberg's Theorem:
\begin{theorem}[\bf Tverberg's Theorem \cite{Tve66}]\label{Tverberg-thm}
let $r,d \geq 1$ be integers. Let $P\subset \Re^d$ be a set of $(r-1)(d+1)+1$ points. Then $P$ admits a partition into $r$ sets  $P_1,\ldots, P_r$ such that $\bigcap_{i=1}^r \textnormal{conv}(P_i) \neq \emptyset$.
\end{theorem}
We need the following well known lemmas (see, cf. \cite{MATOUSEK}):
\begin{lemma}
\label{lex-min}
Let $j,d$ with $j > d$ be some fixed integers. Let $\F = \{ C_{1}, \ldots, C_{j}\}$ be a family of convex sets in $\mathbb{R}^d$ with a nonempty intersection. Then the lexicographic minimum of $\bigcap_{i=1}^{j}C_{i}$ is equal to the lexicographic minimum of the intersection of some $d$-tuple in $\F$.
\end{lemma}

\begin{theorem}[\bf \CARA Theorem]
Let $X\subset \mathbb{R}^{d}$. Then for each $x\in \textnormal{conv}(X)$ there exist $x_1,\ldots, x_{t}\in X$ with $t\leq d+1$ such that $x=\textnormal{conv}(\{x_1,\ldots, x_t\})$.

\end{theorem}
 %Hence, at least one set contains exactly $d+1$ points.

We give a short and simple proof of Theorem~\ref{first-selection} using, essentially, only Tverberg's Theorem and only for $r = d+1$, that is, only for a set of $d^2+d+1$ points.
\begin{proof}[Proof of Theorem~\ref{first-selection}:]
Let $P$ be a set of $n$ points in general position in $\Re^d$. We can assume that $n$ is large enough, otherwise one can choose $c$ sufficiently small and take a point contained in one $d$-simplex. First, we show that for any  $(d^2+d+1)$-tuple $D$ of $P$ there is a set $S\subset D$ of $d+1$ points and a Tverberg partition of $D\setminus S$ into $d$ sets $D_1,\ldots,  D_{d}$ such that the lexicographic minimum of $\bigcap_{i=1}^d \textnormal{conv}(D_{i})$ belongs to conv($S$).

 Tverberg's Theorem implies that any set of $d(d+1)+1 = d^2+d+1$ points can be partitioned into $d+1$ sets $A_{1}, \ldots, A_{d+1}$ all of which convex hulls have a nonempty intersection. Some sets of the partition might be larger than $d+1$, but using \CARA theorem we can assume that each set contains at most $d+1$ points, since one can move points from sets containing at least $d+2$ points to those containing at most $d$ preserving the nonempty intersection property. Using Lemma \ref{lex-min}, we can
assume without loss of generality that the lexicographic minimum $x$ of $\bigcap_{i=1}^{d+1} \textnormal{conv}(A_i)$ is equal to the lexicographic minimum of $\bigcap_{i=1}^{d} \textnormal{conv}(A_i)$. %If $A_{d+1}$ contains $d+1$ points, then setting $S:= A_{d+1}$ and $D_i=A_{i}$ for $1\leq i \leq d$ we are done. Otherwise, some of $A_1, \ldots, A_{d}$ contain $d+1$ points.

 If $x$ is in the interior of conv($A_j$) for a set $A_j$ with $1\leq j\leq d$ containing $d+1$ points, then it is easily seen that $x$ is also the lexicographic minimum of $\bigcap_{\{1\leq i\leq d \}\setminus \{j\}} \textnormal{conv}(A_{i})$. This means that $x$ is the lexicographic minimum of $\bigcap_{\{1\leq i\leq d+1 \}\setminus \{j\}} \textnormal{conv}(A_{i})$ as well, so we can set $S$ to be $A_j$ and take the remaining $d$ sets as the sets $D_i$.

  Hence, one can assume that $x$ is on the boundary of conv($A_j$) for each $A_j$ containing $d+1$ points. This means that for each of these sets containing $d+1$ points it is possible to remove a point so that the lexicographic minimum of the intersection of the convex hulls of the new sets is still $x$.
%iously, there exists $a\in A_j$ such that $x$ is still the lexicographic minimum of $A_{j}\setminus\{a\}\bigcap_{1\leq i\leq d \setminus \{j\}} A_{i}$.
We iteratively eliminate such points from the sets $A_i$ for $1\leq i \leq d$ until the union of the new sets 	$A'_i$ contains $d^2$ points. This is possible, since removing one point from each $A_i$ with $1\leq i \leq d$ containing $d+1$ points leads to $d$ sets whose union contains at most $d^2$ points, while  $\bigcup_{i=1}^d A_i$ contains at least $d^2$ points. If we denote by $A$ the set of eliminated points, then $S:=A\cup A_{d+1}$ contains $d+1$ points and $x\in \textnormal{conv}(A_{d+1})\subset \textnormal{conv}(S)$. Hence, setting $D_i:=A'_i$ we are done.

Next, define $\mathcal{T}$ to be the sets of all Tverberg partitions of $d^2$-tuples of $P$ into $d$ sets.  There are exactly ${n \choose {d^2}}$ such tuples and each of them admits at most some constant $b_d$ of Tverberg partitions into $d$ sets. Hence $|\mathcal{T}|\leq b_{d}{n \choose {d^2}}$. By the reasoning above, each $(d^2+d+1)$-tuple of points from $P$ gives rise to a pair $(\Delta,\mathcal{H})$, where $\Delta$ is a $d$-simplex spanned by $P$ and $\mathcal{H} \in \mathcal{T}$ so that $\Delta$ contains the Tverberg point determined by $\mathcal{H}$. Moreover, a pair $(\Delta, \mathcal{H})$ determines the $(d^2+d+1)$-tuple from which it arises. Hence, the number of pairs $(\Delta,\mathcal{H})$ as above is at least ${n \choose {d^2+d+1}}$. Applying the pigeonhole principle, at least one of the Tverberg points determined by a partition from $\mathcal{T}$ is contained in at least  $$\frac{\binom{n}{d^2+d+1}}{b_{d}\binom{n}{d^2}}=\frac{\binom{n}{d+1}\binom{n-d-1}{d^2}}{\binom{d^2+d+1}{d+1}b_d\binom{n}{d^2}}\geq c \frac{\binom{n}{d+1}}{b_d\binom{d^2+d+1}{d+1}}$$ $d$-simplices, where $c\to 1$ when $n\to \infty$. This completes the proof.

%${n \choose {d^2+d+1}}/b_d {n \choose {d^2}} \geq  f(d){n \choose {d+1}}$ $d$-simplices spanned by $P$ where $f(d) = \Omega((d+1)!/b_d(d^2+d+1)^{d+1})$

%Next, for all $d^2$ tuples of points from $P$ and each Tverberg partition of a $d^2$ tuple, take the Tverberg point this partition determines and call the resulting set $T$. There are exactly ${n \choose {d^2}}$ such tuples and each such tuple admits at most some constant $b_d$ of Tverberg partitions into $d$ sets. Hence $|T|\leq b_{d}{n \choose {d^2}}$. Consider all subsets of $P$ of size $d^2+d+1$. There are ${n \choose {d^2+d+1}}$ such sets. By the reasoning above, each $(d^2+d+1)$-tuple of points from $P$ gives rise to a pair $(\Delta,x)$ where $\Delta$ is a $d$-simplex spanned by $P$ and $x \in T$ so that $\Delta$ contains $x$. Moreover, a pair $(\Delta,x)$ determines the $(d^2+d+1)$-tuple from which it arises.  Hence, by the pigeonhole principle, at least one of the points in $T$ is contained in at least ${n \choose {d^2+d+1}}/b_d {n \choose {d^2}} \geq  f(d){n \choose {d+1}}$ $d$-simplices spanned by $P$ where $f(d) = \Omega((d+1)!/b_d(d^2+d+1)^{d+1})$. This completes the proof.
\end{proof}

\begin{remark}
For $d=2$, this proof can be made truly elementary in opposite to other known proofs. Indeed, the distinct configurations of $7$ points can analyzed by hand. Once this is done, the proof above is reduced to the application of the pigeonhole principle.
\end{remark}

\begin{remark}
In applications, we sometimes want to know that there is a point piercing many $d$-simplices spanned by $P$ in their interior. Lemma 9.1.2 from \cite{MATOUSEK} guarantees that for a set of $n$ points $P$ as in Theorem \ref{first-selection} any point is on the boundary of at most $O(n^d)$ $d$-simplices spanned by $P$. Hence, Theorem \ref{first-selection} guarantees the existence of a point piercing  $\Omega( \binom{n}{d+1})-O(n^d)$ $d$-simplices in their interior.
\end{remark}

\begin{proof}[Proof of Theorem~\ref{rich-simplex}:]
As established in the previous proof, the number of pairs $(\Delta,\mathcal{H})$ is at least ${n \choose {d^2+d+1}}$. Moreover, if $\mathcal{H}_1$ and $\mathcal{H}_2$ are Tverberg partitions of the same $d^2$-tuple, then at most one pair $(\Delta,\mathcal{H}_1)$, $(\Delta,\mathcal{H}_2)$ is counted. Hence, by the pigeonhole principle, for at least one  $d$-simplex $\Delta$ spanned by $P$ there exist at least $$\frac{\binom{n}{d^2+d+1}}{\binom{n}{d+1}}=\frac{\binom{n}{d^2}\binom{n-d^2}{d+1}}{\binom{d^2+d+1}{d+1}\binom{n}{d+1}}\geq c\frac{\binom{n}{d^2}}{\binom{d^2+d+1}{d+1}}$$
$d^2$-tuples of $P$ each having a Tverberg partition  into $d$ sets determining a Tverberg point that belongs to $\Delta$, where $c\to 1$ when $n\to \infty$.
\end{proof}
 %each $(d^2+d+1)$-tuple of points from $P$ gives rise to a pair $(\Delta,\mathcal{H})$, where $\Delta$ is a $d$-simplex spanned by $P$ and $\mathcal{H} \in \mathcal{T}$ so that $\Delta$ contains the Tverberg point determined by $\mathcal{H}$. Since the pair $(\Delta,\mathcal{H})$ determines the $(d^2+d+1)$-tuple from which it arises, by the pigeonhole principle, for at least one  $d$-simplex $\Delta$ spanned by $P$ there exist ${n \choose {d^2+d+1}}/{n \choose {d+1}} =\Omega({n \choose {d^2}})$ Tverberg partitions of $d^2$-tuples into $d$ sets  determining Tverberg points that belongs to $\Delta$.
%This completes the proof. Note that if $x_1$ and $x_2$ arise from the same $d^2$ tuple (but different partitions) at most one of $(\Delta,x_1)$ and $(\Delta,x_2)$ is counted, otherwise the same $(d^2+d+1)$-tuple is considered twice.

\section{Weak $\frac{1}{r}$-net in a Kinetic Setting}

 In a kinetic setting, one needs to capture the combinatorial changes occurring with time. The concept of a \emph{kinetic hypergraph} defined below was introduced in \cite{Ca15} by De Carufel et al.
\begin{definition}[\bf Kinetic Hypergraph]
\it{Let $P$ be a set of points moving in $\mathbb{R}^d$ with bounded description complexity and let $\mathcal{R}$ be a set of ranges. We denote by $(P,\mathcal{S})$ the kinetic hypergraph of $P$ with respect to $\mathcal{R}$. Namely, $S \in \mathcal{S}$ if and only if there exists an $R \in \mathcal{R}$ and a "time" $t \in \mathbb{R}_{+}$ such that $S(t)= R\cap P(t)$. We sometimes abuse the notation, and denote by $(P,\mathcal{R})$ the kinetic hypergraph $(P,\mathcal{S})$.
 }
\end{definition}
Figure 1 illustrates the concept of a kinetic hypergraph for $d=1$ and for $\mathcal{R}$ the family of intervals. De Carufel et al. \cite{Ca15} also established the following important lemma to investigate strong $\frac{1}{r}$-nets in a kinetic setting.

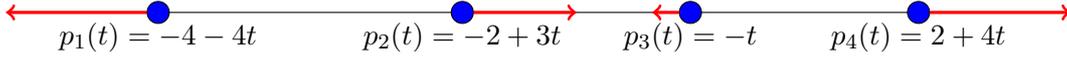
\begin{figure}[t]
\label{fig}
\begin{center}
\begin{tikzpicture}[scale=2]
\draw (-5,0) -- (2,0);

\node[shape=circle, draw, fill=blue,scale=0.8] (critical) at (-4,0){};
\node[below,scale=1]  at (-4,0){$p_1(t)= -4-4t$};
\draw[->,red, very thick]  (critical) to (-5,0) ;

\node[shape=circle, draw, fill=blue,scale=0.8] (britical) at (-2,0){};
\node[below,scale=1]  at (-2,0){$p_2(t)= -2+3t$};
\draw[->,red, very thick]  (britical) to (-1.25,0) ;

\node[shape=circle, draw, fill=blue,scale=0.8] (britical) at (-0.5,0){};
\node[below,scale=1]  at (-0.5,0){$p_3(t)=-t$};
\draw[->,red, very thick]  (britical) to (-0.75,0) ;

\node[shape=circle, draw, fill=blue,scale=0.8] (britical) at (1,0){};
\node[below,scale=1]  at (1,0){$p_4(t)= 2+4t$};
\draw[->,red, very thick]  (britical) to (2,0) ;

%\node[fill=blue!60,shape=circle, draw,scale=0.6] at (6,2){$p_1$};
%\node[fill=blue!60,shape=circle, draw,scale=0.6] at (8,2){$p_2$};
%\node[fill=blue!60,shape=circle, draw,scale=0.6] at (6,0){$p_3$};
%\node[fill=blue!60,shape=circle, draw,scale=0.6] at (8,0){$p_4$};
%\draw[red] plot[smooth, tension=.7] coordinates {(5.5,0) (5.8,2.5) (8.7,3) (7.9,1) (6,-1) (5.5,0)};
%\draw[red] plot[smooth, tension=.7] coordinates {(7.5,0) (8.3,-1) (9,0) (9.4,2) (9.1,3) (5.3,0)};

\end{tikzpicture}
\end{center}
\caption{A family $P$ moving linearly along the the real line. One can easily see that the kinetic hypergraph of $P$ with respect to intervals is ($P$,$2^{P}\setminus \{\{p_1,p_2,p_4\},\{p_1,p_3,p_4\}, \{p_1,p_4\}\}$).}
\end{figure}

 \begin{lemma}[\bf De Carufel et al. \cite{Ca15}]
 \label{VC}
 \it{Let $\mathcal{R}$ be a collection of semi algebraic sets in $\mathbb{R}^d$, each of which can be expressed as a Boolean combination of a constant number of polynomial equations and inequalities of maximum degree $c$, where $c$ is some constant. Let $P$ be a family of points moving polynomially in $\mathbb{R}^d$ with bounded description complexity. Then the kinetic hypergraph of $P$ with respect to $\mathcal{R}$ has bounded \VC-dimension.}
 \end{lemma}

Unfortunately, this is not enough for our purposes since we need to assume that the moving points can be described with coordinates which are rational functions. However, following a similar scheme it is not hard to prove the lemma below:
\begin{lemma}
\label{VC-INT}
 \it{Let $P$ be a set of moving points in $\mathbb{R}$ with bounded description complexity, and let $\mathcal{K}=(P, \mathcal{S})$ be the kinetic hypergraph on $P$ with respect to intervals. Then the \VC-dimension of $\mathcal{K}$  is $O(1)$.}
 \end{lemma}

First, we need the following link between the \emph{primal shatter function} and the \VC-dimension of a set system.
\begin{definition}
\it{For a set system $X=(P,\mathcal{S})$ the primal shatter function $\pi_{X}:\{1, \ldots, |P|\}\to \mathbb{N}$ is defined by $$\pi_{X}(m)= \max_{A\subset P :~ |A|=m}|\{ A\cap S : S \in \mathcal{S}\}|.$$}
\end{definition}
 The lemma below is a folklore:

\begin{lemma}
\label{VC-A}
\it{Let $X=(P,\mathcal{S})$ be a set system such that $\pi_{X} (m) \leq cm^k$ (for $k\geq 2$ say), where $c$ is some constant, and let $d$ be the \VC-dimension of $X$. Then $d=O(k\log k)$.}
 \end{lemma}

\begin{proof}
If $d=0$ then there is nothing to show. Otherwise, $\pi_{X}(d)$ is defined, and we easily see that $c\geq 2$ since $\pi_{X}(1)=2$. Hence, by the definition of the primal shatter function and the lower bound on $c$, the following inequalities are satisfied $2^d \leq cd^k \leq (cd)^k$. This implies that $d\leq k\log cd$. Obviously, there is a $c'>0$ (depending only on $c$) such that $$c'd^{\frac{1}{2}}\leq \frac{d}{\log cd}\leq k.$$
Hence, $$d\leq k\log \frac{c}{{c'}^2}k^2=k\log \frac{c}{{c'}^2}+2k\log k=O(k\log k).$$

\end{proof}

%In order to establish our result, we use a framework inspired from De Carufel et al. \cite{Ca15}.

The following definition provides a natural tool to count the number of hyperedges in a kinetic hypergraph.
\begin{definition}

\it{Given a set $I \subset \mathbb{R}_+$, we define the restriction of a kinetic hypergraph $(P, \mathcal{R})$ to $I$, and denote it by $(P, \mathcal{R})|_{I}$. It is the  hypergraph $(P,\mathcal{E}')$
 where $$\mathcal{E}':=\{P'\subset P: \exists t\in I~\textnormal{and}~ R\in \mathcal{R}~\textnormal{with}~ P'(t)=P(t)\cap R\}.$$}
 \end{definition}
Assume that $\mathcal{I}$ is a partition of $\mathbb{R}_{+}$. If $m_{I}$ denotes the number of hyperedges in $(P,\mathcal{R})|_{I}$, then clearly the number of hyperedges in $(P,\mathcal{R})$ is at most $\sum_{I\in\mathcal{I}}m_{I}$. If $t>0$ then we say that $p_i,p_j$ become \emph{coincident} at $t$ if $p_i(t)=p_j(t)$ and for some $\delta >0$ we have $p_i(\tilde{t})\neq p_j(\tilde{t})$ for $\tilde{t} \in [t-\delta, t[ $. It is easy to see that in the case of motion with bounded description complexity if $p_i(t)=p_j(t)$, but $p_i$ and $p_j$ do not become coincident at $t$, then $p_i(t)=p_j(t)$ for all $t \in \mathbb{R}_{+}$. For technical reasons, the lemma below holds without any general position assumption.

\begin{lemma}
\label{NHE}
\it{Let $P$ be a set of $n\geq 1$ points moving in $\mathbb{R}$ with bounded description complexity $\beta$. Then the number of hyperedges in the kinetic hypergraph $\mathcal{K}=(P,\mathcal{S})$ with respect to intervals is at most $(10\beta +4) n^4$.}
\end{lemma}

\begin{proof}

Let $p_{1}, \ldots,  p_{n}$ be the points in $P$. First, we show that if a hyperedge  appears at some $t$, then one of the following conditions holds:
\begin{itemize}
\item[i)]$t=0$,
\item[ii)]at least 2 points  become coincident at $t$,
\item[iii)] one of the moving points is not defined at $t$.
\end{itemize}

Let us assume that when a hyperedge $S$ appears all points are defined and $t\neq 0$. This implies that their motion is continuous  on $[t-\delta, t+ \delta]$ for a $\delta >0$ small enough.  %By definition of $t$, there should be a sequence $\{t_m\}_{m=0}^{\infty}$ with $t_m \geq t$, $\lim_{m\to \infty}t_m=t$ and a sequence $\{b_m\}_{m=0}^{\infty}$ such that $P(t_{m})\cap ]-\infty,b_m]= S(t_{m})$ or
%$P(t_{m})\cap [b_m, \infty[= S(t_{m})$. W.l.o.g. we can assume that $b_m \cap S(t_{m})\neq \emptyset$. Hence, using the Bolzano-Weierstrass lemma one can extract a converging subsequence $\{b_{m_i}\}_{i=0}^{\infty}$ with limit $b$ such that w.l.o.g. $P(t_{m_i})\cap ]-\infty,b_{m_i}]= S(t_{m_i})$. It is not hard to see that $]-\infty, b[ \cap P(t) \subset S(t)$ and $]b, \infty[ \cap P(t) \subset S^c(t)$.
%Indeed, assume by contradiction that some $x(t)\in ]-\infty,b[\cap P(t)$ is in $S^c(t)$.  Then $x(t_{m_i})> b_{m_i}$ implying $$x(t)= \lim_{i\to \infty} x({t_{m_i}})\geq \lim_{i\to \infty} b_{n_{i}}=b$$ a contradiction. The proof that $]b, \infty[ \cap P(t) \subset S^c(t)$ is analogous. Hence, $$]-\infty, b[ \cap P(t) \subset S(t) \subset  ]-\infty, b] \cap P(t).$$
For a hyperedge $E$ present at some time $t'$, there is an interval $[p_{i}(t'),p_{j}(t')]$ such that $[p_{i}(t'),p_{j}(t')]\cap P(t')=E(t')$. Clearly, if a hyperedge appears at $t$, a moving point from $P$ must enter or leave an interval as above at $t$. Hence, two points become coincident implying that we are indeed in one of the three cases above.

%By the bounded description complexity assumption, both sides are polynomials of degree at most $2\beta$.
%If no pair of points become coincident at $t$, then by the observation preceding this lemma the points located in $b$ at the time $t$ must have the same motion. Hence, all these points are in $S$ or all of them are outside, so by continuity $S$ appears before $t$, a contradiction. Hence we are in one the three cases.\\
Let $0=t_0<t_1<t_2\ldots$ be the sequence of times %for which the position of one of the $x_i$ is not defined or at least $2$ moving points become coincident.
 corresponding to one of three cases listed above. We refer to these times as \emph{special events}. In the second case, if $p_{i}(t)=a_{i}(t)/b_{i}(t)$ and $p_{j}(t)=a_{j}(t)/b_{j}(t)$ become coincident, then  $a_{i}(t)b_{j}(t)=a_j(t)b_i(t)$.  As we already observed, the polynomials $a_{i}(x)b_{j}(x)$ and $a_j(x)b_i(x)$ are not identical (otherwise $p_i(t)=p_j(t)$ for all $t\geq 0$) and of degree at most $2\beta$. Hence, the equality is satisfied for at most $2\beta$ values of $x$.  Since the number of special events where at least one of the $p_j$ is not defined is clearly at most $n\beta$, it follows that the total number of special events is at most $1+n\beta + \binom{n}{2}2 \beta$.

  It is not hard to see that the number of hyperedges in the static case is at most $\binom{n}{2}+n+1$ (one should not forget $\{\emptyset\}$!). Since no special event can occur  in the interval $]t_i,t_{i+1}[$ (where $t_{i+1}=\infty$ if $i$ is the index of the last special event), it follows from the reasoning above that the number of hyperedges in $\mathcal{K}|_{]t_i, t_{i+1}[}$ is at most $\binom{n}{2}+n+1$ as well. Indeed, if the restricted hypergraph contains at least $\binom{n}{2}+n+2$ hyperedges, then there is an interval $[a,b]\subset]t_i,t_{i+1}[$ such that $\mathcal{K}|_{[a,b]}$ also contains at least $\binom{n}{2}+n+2$ hyperedges. However, considering the motion of points of $P$ starting from $a$ and applying the reasoning above, we conclude that the hyperedges of $\mathcal{K}|_{[a,b]}$ may appear only at $a$, hence their number is at most $\binom{n}{2}+n+1$, a contradiction.

  Since the total number of intervals, of zero and non-zero length, in the partition given by the special events is at most $2(n\beta + 1 +\binom{n}{2}2 \beta)$, the total number of hyperedges in the kinetic hypergraph is at most $$2(n\beta + 1 +\binom{n}{2}2 \beta)(\binom{n}{2}+n+1)\leq (10\beta +4) n^4$$
as asserted.
\end{proof}

%\begin{remark}
%The lemma above is also valid if the points are allowed to not be defined. That is, if we allow $x\in X$ to be the map $x: \mathbb{R}_{+}\to \emptyset$, since such points do not contribute to the number of edges. We will use this observation in the proof of the main theorem.
%\end{remark}

It is easy to see that the bound on the number of hyperedges above is also valid for induced hypergraphs.  Consider an induced hypergraph $(X,\mathcal{S}|_{X})$ of $ (P, \mathcal{S})$, and let $A=S\cap X$ be a hyperdge of  $(X,\mathcal{S}|_{X})$ arising from some $S\in \mathcal{S}$. By definition, there is an interval $[a,b]$ and a $t\geq 0$ such that $P(t)\cap [a,b]=S(t)$. We now show that $[a,b]\cap X(t)= A(t)$. Clearly, $A(t)\subset [a,b]$ otherwise for some $a\in A$ we have $a(t)\not\in S(t)$ implying $a \not \in S$, hence $A(t)\subset [a,b]\cap X(t)$.

Let us prove $[a,b]\cap X(t)\subset A(t)$. Take an $x(t)\in X(t)\cap [a,b]$, then clearly $x \in S$ implying $x\in S\cap X= A$, so $x(t) \in A(t)$. This proves that the induced hypergraph $(X, \mathcal{S}|_{X})$ is contained in the kinetic hypergraph of $X$ with respect to intervals, hence the bound of Lemma \ref{NHE} holds for induced hypergraphs that have at least one vertex.

\begin{proof}[Proof of Lemma \ref{VC-INT}]
The lemma is an immediate corollary of Lemma~\ref{VC-A} combined with Lemma~\ref{NHE}.
\end{proof}

Combined with the well known strong $\frac{1}{r}$-net theorem mentioned in Section~\ref{intro}, Lemma \ref{VC-INT} implies:

%For a fixed $h$ as above, we denote by $\mathcal{F}_h$ the set of all moving points arising from some $(j+1)$-tuple of $P$ and $h$.
\begin{lemma}
\label{Lemma3}
 \it{Let $P$ be a set of points moving in $\mathbb{R}$ with bounded description complexity. Then the kinetic hypergraph of $P$ with respect to intervals has a strong $\frac{1}{r}$-net (for $r\geq 2$ say) of size $O(r\log r)$.}
\end{lemma}
It is important to note that the lemma above holds without any general position assumption. Hence, for any $t\geq 0$ more than two moving points from $P$ can coincide at $t$. Later on, we shall use Lemma \ref{Lemma3} in order to find weak $\frac{1}{r}$-nets in a kinetic setting.

The proof of Theorem~\ref{MT} below is inspired by a construction from Chazelle et al. \cite{Ch95}. We recall the general position assumption made in Section~\ref{intro}: Given a set of moving points $P$ in $\mathbb{R}^d$, for any $t\geq 0$ the affine hull of any $d$-tuple of points in $P(t)$ is a hyperplane, and no $d+2$ points of $P(t)$ are contained in a hyperplane. All arguments of the theorem below are also valid when the set $P$ consist of points with bounded description complexity. However, as explained in the first section, when the motion is polynomial the construction we present has an important feature: one coordinate is a polynomial. In particular, when $d=2$ the construction below gives a kinetic weak $\frac{1}{r}$-net $N$ of size only $O(r^{3}\log^2 r)$ and the first coordinate of each point in $N$ is a polynomial. Note that in the static setting, the best known upper bound on the function $f(2,r)$, defined in Section~\ref{intro}, is $O(r^2)$, so our bound is only an $O(r\log^2 r)$ factor of it.

\begin{theorem}[\bf Weak $\frac{1}{r}$-net in a Kinetic Setting]
\label{MT}
\it{Let $P$ be a set of $n$ points moving polynomially in $\mathbb{R}^d$ with bounded description complexity $\beta$. Then there exists a kinetic weak $\frac{1}{r}$-net (for $r\geq 2$ say) $N$ of size $O(r^{\frac{d(d+1)}{2}}\log^{d} r)$ and bounded description complexity. Moreover, the first coordinate of each point of $N$ is a polynomial.}
\end{theorem}
\begin{proof}

The case $d=1$ is implied by Lemma \ref{Lemma3}, so we can assume that $d\geq 2$. The method below works for $n\geq cr$, where $c$ is a sufficiently large constant whose existence is proved later. If $n< cr$ then the theorem holds trivially, since one defines the kinetic weak $\frac{1}{r}$-net to be $P$.

We start by defining $N$ and other structures we need throughout the proof. Later, we show that $N$ is indeed a kinetic weak $\frac{1}{r}$-net for $P$. The claims regarding the size and the description complexity of $N$ will follow easily from its definition. First, we need to introduce the concept of a \emph{moving space of step $j$} for $1\leq j \leq d$. It will be some specific moving space of dimension $d-j$. Moreover, a moving space of step $i+1$ arises from some moving space of step $i$, hence these structures will be defined iteratively. In what follows, we use parameters $\lambda_{1},\ldots, \lambda_{d}$ with $0<\lambda_i\leq 1$, whose values are specified later.

Call the projection of $P$  onto $x_1$-axis $P_1$. Note that $P_1$ has description complexity $\beta$. Choose a strong $\frac{\lambda_1}{r}$-net $N_1$ for the kinetic hypergraph of $P_1$ with respect to intervals. Lemma \ref{Lemma3} guarantees that one can select $N_1$ with $|N_1|\leq b_1 r/\lambda_1\log r/\lambda_1$, where $b_1$ depends on $\beta$. For each point $p$ of $N_1$, we consider the moving hyperplane such that at any $t\geq 0$ it is orthogonal to $x_1$-axis and passes through $p(t)$. The moving affine spaces of step $1$ are exactly these moving hyperplanes arising from $N_1$.

 The construction of the moving spaces of step at least 2 is more involved. Assume that we have constructed the moving affine spaces up to step $j$ satisfying  $1 \leq j\leq d-1$. For each moving space $h$ of step $j$, we define $F_h$ to be the set consisting of moving points $p^{h,X}$ for all $(j+1)$-tuples $X$ of $P$. The position of  $p^{h,X}$ at $t\geq 0$ is given by $p^{h,X}(t)=\textnormal{aff}(X(t))\cap h(t)$ if this intersection contains a single point. A moving point $p^{h,X}$ is not necessarily uniquely defined, but this not a problem for our purposes. One can define it with description complexity $f(j+1)$ for some increasing function $f : \{1,\ldots, d\} \to \mathbb{N}$ such that $f(1)=\beta$. The technical proof of this fact is provided later in Lemma \ref{TP}. Next, for each moving space $h$ of step $j$ call the projection of $F_h$ onto $x_{j+1}$-axis $P_h$. Note that $P_h$ also has description complexity $f(j+1)$. Choose a strong $\frac{\lambda_{j+1}}{r^{j+1}}$-net $N_{h}$ for the kinetic hypergraph of $P_h$ with respect to intervals. Again, Lemma \ref{Lemma3} ensures that one can select $N_{h}$ with $|N_{h}|\leq b_{j+1}r^{j+1}/\lambda_{j+1} \log r^{j+1}/\lambda_{j+1}$, where $b_{j+1}$ depends on $f(j+1)$. If $N_{h}$ consists of $q_{1},\ldots, q_{s}$, then the moving affine spaces of step $j+1$ induced by $h$  are $\tilde{h}_{i}$ given by $x_1=x_{h,1}, \ldots ,~x_j=x_{h,j},~x_{j+1}=q_{i}$ for $1\leq i \leq s$, where $x_{h,k}$ is the moving point giving the $k^{th}$ coordinate of $h$. The set of moving spaces of step $j+1$ is the union of moving spaces induced by $h$ among all moving spaces $h$ of step $j$.
 %Given $p_{h,1},\ldots, p_{h,m}$ induced by $h$ as above
% This makes sense since the moving spaces constructed at step $d$ are $0$-dimensional, so they are moving points.

We define the kinetic weak $\frac{1}{r}$-net $N$ to be the union of the moving spaces of step $d$.  This makes sense, since the moving spaces of step $d$ have each coordinate specified by some function, so those are moving points. The size of $N$ is at most
$$b_{1}\frac{r}{\lambda_1}\log\frac{r}{\lambda_1} b_{2}\frac{r^2}{\lambda_2}\log \frac{r^2}{ \lambda_2} \ldots b_{d}\frac{r^d}{\lambda_d} \log \frac{r^d}{ \lambda_d}=O( r^{\frac{d(d+1)}{2}}\log^d r).$$
Moreover, for each $s=(s_1,\ldots, s_d)$ of $N$, the moving point $s_i$ has description complexity $f(i)$. Since $f$ is an increasing function, the moving point $s$ has description complexity $f(d)$.

% We show by induction that one can define $\lambda_i$ for $1\leq i\leq d-1$, so that there is a moving space  $h$ of step $i$ such that $h(t)$ pierces a lot of $i$-simplices spanned by $C\cap P(t)$

First, we briefly outline the the main ideas of the proof for $d\geq 3$. The case $d=2$ is much easier, and does not require the inductive step presented below.  For the sake of simplicity, some arguments below are not entirely correct.

 Let $t\geq 0$ and let $C$ be a convex set containing $>n/r$ points of $P(t)$. We start by showing that if one chooses an appropriate value for $\lambda_1$, then for some moving space $h$ of step $1$ the set $h(t)$ intersects ``many" segments spanned by $C\cap P(t)$.

  Next comes the inductive step. We assume that $\lambda_i$ were defined up to some $1\leq j\leq d-2$, and some moving space $h$ of step $j$ (of dimension $d-j$) is such that $h(t)$ intersects a ``large" number of $j$-simplices spanned by $C\cap P(t)$. We first find a static affine space $s$ contained in $h(t)$ of dimension $d-j-1$ such that $s$ intersects a ``large" number of $(j+1)$-simplices spanned by $C\cap P(t)$. These $(j+1)$-simplices are obtained from the $j$-simplices intersecting $h(t)$.  Then we show that with an appropriate choice of $\lambda_{j+1}$, there are two  moving spaces $h_1$, $h_2$ of step $j+1$  induced by $h$ such that $h_1(t)$ and $h_2(t)$ are ``close" to $s$ and therefore at least one of them also intersects a ``large" number of $(j+1)$-simplices spanned by $C\cap P(t)$, which completes the inductive step.

  This way, we establish that one can define $\lambda_i$  for $1\leq i\leq d-1$, so that for some moving line $l$ of step $d-1$ ``many" $(d-1)$-simplices spanned by $C\cap P(t)$ are intersected by $l(t)$. In particular, from the definition of $F_l$ the segment $C\cap l(t)$  is such that for ``many" moving points $p\in F_l$ the point $p(t)$ belongs to it. Hence, the projection of $C\cap l(t)$ (call it $I$) onto $x_d$-axis leads to a ``heavy" hyperedge in the kinetic hypergraph of $P_l$ with respect to intervals (because $P_l$ is the projection of $F_l$). For an appropriate choice of $\lambda_d$ there is a point $q$ of the net $N_l$ such that $q(t)$ must be in $I$. Finally, by construction of $N$ the moving point whose first $d-1$ coordinates are given by $l$ and the last one by $q$ is in $N$, so $q(t)$ is in $C$ and we are done.

We now proceed with a detailed proof. Let us show that the set $N$ we defined is indeed a kinetic weak $\frac{1}{r}$-net for $P$ for an appropriate choice of $\lambda_i$. Let $t\geq 0$ and let $C$ be any convex set containing at least $n/r$ points from $P(t)$. It is sufficient to assume that $C$ contains exactly $n/r$ points of $P(t)$ (we choose any $n/r$ points of $C\cap P(t)$, and disregard the remaining ones). We will define the parameters $\lambda_i$ so that $C$ must contain a point of $N(t)$. It is important to notice that these parameters do not depend on $C$ or $t$. For technical reasons,  for $1\leq j\leq d-1$ we also prove the existence of $\gamma_j n^{j+1}/r^{j+1}$  $j$-simplices spanned by $C\cap P(t)$ and intersecting  $h(t)$ for some moving space $h$ of step $j$ exactly once in their relative interior, where $\gamma_j>0$ are iteratively defined later. Clearly, this implies for each simplex above that the affine hull of the $j+1$ points of $P(t)$ spanning it intersects $h(t)$ exactly once as well. In particular, if $h(t)$ intersects $\gamma_j n^{j+1}/r^{j+1}$  $j$-simplices spanned by $C\cap P(t)$ once in their relative interior, then for at least  of $\gamma_j n^{j+1}/r^{j+1}$ points $p\in F_h$ we have $p(t)\in C\cap h(t)$. This implication is crucial for our purposes, and will be used in order to prove that $N$ is a kinetic net once the parameters $\lambda_i$ are specified. We prove the existence of $ \gamma_j$ and define $\lambda_j$ for $1\leq j\leq d-1$ by induction. Then, we define $\lambda_d$ and show that the values $\lambda_j$ imply that $N$ is a kinetic weak $\frac{1}{r}$-net.
\begin{lemma}
\label{T1}
If $\lambda_1=1/4$ and $n> 4r(2d+2)$, then there exists a moving hyperplane $h$ of step 1 such that $h(t)$ intersects at least $n^2/16r^2$ segments spanned by  $C\cap P(t)$ once and in their relative interior.
\end{lemma}
\begin{proof}
%For $j=1$ set $\lambda_1=1/4$ and assume $n> 4r(2d+2)$.

 Among moving spaces $\tilde{h}$ of step 1 with the property  that  $> n/4r$ points of $C\cap P(t)$ have a strictly smaller $x_1$-coordinate than the intersection of $\tilde{h}(t)$ with $x_1$-axis, choose a moving space with the smallest intersection point with $x_1$-axis at $t$ and denote it by $h$. The moving space $h$ exists, since the above defined set of moving spaces is easily seen to be nonempty.

 Indeed, let $z$ be the largest real such that at most $n/4r$  points of $C\cap P(t)$ have their $x_1$-coordinate in $]-\infty, z[$. Then, since from the general position assumption at most $d+2$ points of $C\cap P(t)$ can share the same $x_1$-coordinate, we deduce that there are at least $n/r-n/4r-d-2> n/4r$ points of $C\cap P(t)$ whose $x_1$-coordinate is in $]z, \infty[$. Since $N_1$ is a strong $\frac{1}{4r}$-net for the kinetic hypergraph of $P_1$ with respect to intervals, there should be a point $w\in N_1$ such that $w(t)\in ]z, \infty[$ implying the existence of a moving space of step 1 whose intersection with $x_1$-axis at $t$ is $w(t)$. Hence, the above set of moving spaces is indeed nonempty, so $h$ exists.

  Let $x$ denote the intersection point of $h(t)$ with $x_1$-axis. We now show that we also have at least $n/4r$ points of $C\cap P(t)$ having a strictly bigger $x_1$-coordinate than $x$. Indeed, using one more time the hypothesis that no $d+2$ points are contained in a hyperplane, we deduce that the number of points of $C\cap P(t)$ having their $x_1$-coordinate smaller or equal to $x$ is at most $2n/4r + 2(d+1)< 3n/4r$.

  To see this, let $\tilde{h}(t)$ be a predecessor of $h(t)$, i.e., a moving hyperplane of step 1 at $t$ whose intersection point with $x_1$-axis is the biggest one among those having an intersection point with  $x_1$-axis strictly smaller than $h(t)$. The existence of such a hyperplane is again implied by the the definition of $N_1$. Indeed, we have  $>n/4r$ points $p\in P_1$ such that  $p(t)\in ]-\infty,x[$, so there is a point $w\in N_1$ such that $w(t)\in ]-\infty,x[$. Thus, there is moving hyperplane of step $1$ with its $x_1$-coordinate equal to $w(t)$, which implies the existence of $\tilde{h}(t)$. Similarly, it is easily seen that at most $n/4r$ points of $C\cap P(t)$ are strictly between $h(t)$ and $\tilde{h}(t)$. In summary, by the choice of $h$, at most $n/4r$ points of $C\cap P(t)$ have their $x_1$-coordinate strictly smaller than the intersection of $\tilde{h}(t)$ with $x_1$-axis, at most $n/4r$ are strictly between $h(t)$ and $\tilde{h}(t)$, and at most $d+1$ points lie on each of $h(t), \tilde{h}(t)$.

  Hence, both open halfspaces delimited by $h(t)$ contain $>n/4r$ points of $C\cap P(t)$. This means that at least $n^2/16r^2$ segments spanned by  $C\cap P(t)$ intersect $h(t)$ once and in their relative interior, so the lemma follows.
\end{proof}
%Here $2(d+1)$ bounds the number of points on $h(t)$ and on the previous moving hyperplane at $t$, i.e., the one with the biggest intersection point with the $x_1$-axis strictly smaller than $p$.

\begin{figure}[t]
\label{Fig}
\begin{center}
\begin{tikzpicture}[scale=0.7]

\draw (-8,0) -- (5.5, 0);

\filldraw[fill=blue!50,fill opacity=0.6] (-5.8, 7.3) -- (-5.4, 8) -- (-4.6,8.5) -- (-3.5,8.9)  -- (-2.4, 9) -- (0,8.6) -- (1.8,8.2)  --(3.2,4.9) -- (0.8,3.8) -- (-0.8,3.6) -- (-3.5, 4.5) -- (-5,5.6) -- cycle;
\node[color=blue,scale=1.5] at (-4.5,9) {C};

\draw[red]  (-6,0) -- (-6,10);
\node[shape=circle, draw, fill=black, scale=0.4] at (-6,5)  {};

\coordinate (A) at (-5.5,7.6);
\coordinate (B) at (-4.8,6.2);
\coordinate (C) at (-3.8,7.2);
\coordinate (D) at (-3.4,6);
\coordinate (E) at (-3.2,7.1);
\node[shape=circle, draw, fill=black, scale=0.4] at (-5.5,7.6)  {};
\node[shape=circle, draw, fill=black, scale=0.4] at (-4.8,6.2)  {};
\node[shape=circle, draw, fill=black, scale=0.4] at (-3.8,7.2)  {};
\node[shape=circle, draw, fill=black, scale=0.4] at (-3.4,6)  {};
\node[shape=circle, draw, fill=black, scale=0.4] at (-3.2,7.1)  {};
\coordinate (a) at (-6,8);
\coordinate (b) at (-6,3);

\node[shape=circle, draw, fill=red, scale=0.4] at (a)  {};
\node[shape=circle, draw, fill=red, scale=0.4] at (b)  {};
\node[shape=circle, draw, fill=red, scale=0.4] at (-3,3.4)  {};
\node[shape=circle, draw, fill=red, scale=0.4] at (-3,1)  {};
\draw[red]  (-2,0) -- (-2,10);
\coordinate (F) at (-2.8,5.3);
\coordinate (G) at (-2.7,8.7);
\coordinate (H) at (-2.5,7.9);

\node[shape=circle, draw, fill=black, scale=0.4] at (-2,9.5)  {};

\node[shape=circle, draw, fill=black, scale=0.4] at (F)  {};
\node[shape=circle, draw, fill=black, scale=0.4] at (G)  {};
\node[shape=circle, draw, fill=black, scale=0.4] at (H)  {};
%\path [name path=cd, red, dashed, draw] (H) -- (E);

\path[name path=fg, red, draw]  (-3,0) -- (-3,10);
\path[name path=ab, red, draw]  (-2,0) -- (-2,10);

\coordinate (I) at (-1.3,5.6);
\coordinate (J) at (-1.1,8.2);
\coordinate (K) at (-0.6,6.8);
\coordinate (L) at (-0.1,7.6);
\coordinate (M) at (0,5.2);
\path[name path=sk, red, dashed, draw]   (H) -- (J);
\draw [name intersections={of=ab and sk, by=t}] node[shape=circle,fill=red,draw, scale=0.4] at (t) {};
\node[shape=circle, draw, fill=black, scale=0.4] at (I) {};
\node[shape=circle, draw, fill=black, scale=0.4] at (J) {};
\node[shape=circle, draw, fill=black, scale=0.4] at (K) {};
\node[shape=circle, draw, fill=black, scale=0.4] at (L)  {};
\node[shape=circle, draw, fill=black, scale=0.4] at (M) {};

\node[shape=circle, draw, fill=black, scale=0.4] at (-2.5,1.4)  {};
\node[shape=circle, draw, fill=black, scale=0.4] at (4,2.2)  {};
\node[shape=circle, draw, fill=black, scale=0.4] at (3,3.2)  {};
\node[shape=circle, draw, fill=black, scale=0.4] at (2.8,7.2)  {};

\node[shape=circle, draw, fill=black, scale=0.4] at (1.5,8.8)  {};

\draw[red]  (0.5,0) -- (0.5,10);

\node[shape=circle, draw, fill=black, scale=0.4] at (0.5,2.3)  {};

\coordinate (N) at (0.7,6.5);
\coordinate (O) at (0.8,5.4);
\coordinate (P) at (1.5,7.2);
\coordinate (Q) at (1.9,4.7);
\node[shape=circle, draw, fill=black, scale=0.4] at (N) {};
\node[shape=circle, draw, fill=black, scale=0.4] at (O)  {};
\node[shape=circle, draw, fill=black, scale=0.4] at (P) {};
\node[shape=circle, draw, fill=black, scale=0.4] at (Q)  {};
\coordinate (t) at (0.5,5.9);
\node[shape=circle, draw, fill=red, scale=0.4] at (t)  {};
\draw[red,dashed] (K) -- (P);
\coordinate (j) at (-2,6.533);
\draw[red,dashed] (P)-- (j);
\node[shape=circle, draw, fill=red, scale=0.4] at (j)  {};

\draw[red]  (2.1,0) -- (2.1,10);
\coordinate (R) at (2.3,6.7);
\coordinate (S) at (2.8,5.3);
\node[shape=circle, draw, fill=black, scale=0.4] at (R) {};
\node[shape=circle, draw, fill=black, scale=0.4] at (S)  {};

\coordinate (k) at (2.1,8);
\node[shape=circle, draw, fill=red, scale=0.4] at (k) {};
\draw[red]  (3.5,0) -- (3.5,10);
\node[shape=circle, draw, fill=black, scale=0.4] at (2.8,5)  {};
\node[shape=circle, draw, fill=black, scale=0.4] at (-2.3,4.7)  {};
\coordinate (m) at (3.5,4);
\node[shape=circle, draw, fill=red, scale=0.4] at (m) {};
\node[red, below] at (-2,0) {$h(t)$};

\end{tikzpicture}
\end{center}
\caption{The line $h(t)$ splits $C\cap P(t)$ into two parts of cardinality $>n/4r$. At least one point from the net $N$ induced by $h$ must be in $C$ at ``time" $t$.}
\end{figure}
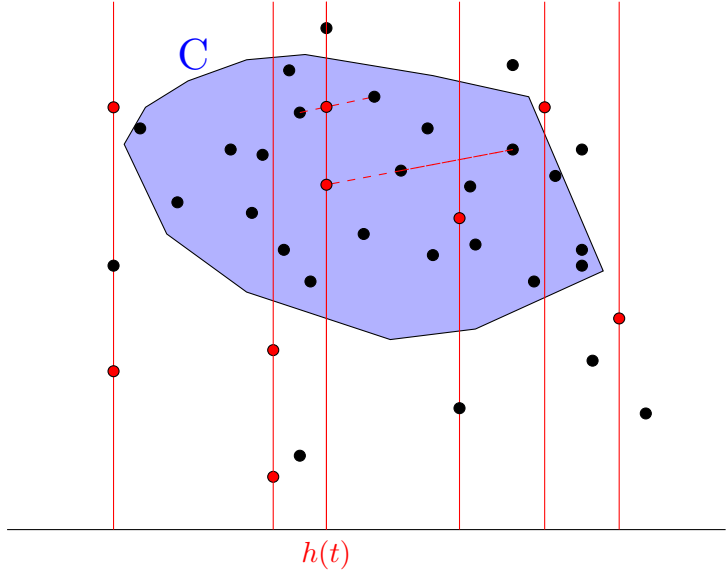

The lemma above implies that if we define $\lambda_1= 1/4$, then we can set $\gamma_{1}=1/16$. If $d=2$ then set $\lambda_2= 1/8$. Indeed, by definition of $F_{h}$ and Lemma \ref{T1} there exist at least $n^2/16r^2> \binom{n}{2}/8r^2=|F_{h}|/8r^2$ points $p$ of $F_h$ such that $p(t)\in C\cap h(t)$.
Since $P_{h}$ is the projection of $F_{h}$ onto $x_2$-axis, there exist $>|P_h|/8 r^2$  points $p'$ of $P_h$ such that $p'(t)$ belongs to the projection of the segment $C\cap h(t)$ onto $x_2$-axis. Hence, since $N_h$ is a strong $\frac{1}{8r^2}$-net for the kinetic hypergraph of $P_h$ with respect to intervals, the projection of $C\cap h(t)$ onto $x_2$-axis must contain a point $v(t)$ of $N_{h}(t)$. By definition of $N$, the moving point $q=(x_{h,1},v)$ is in $N$, so $C\cap N(t) \neq \emptyset$ and the case $d=2$ follows, see Figure 2 for an illustration. Hence, one can assume that $d\geq 3$.

In higher dimensions the analysis requires more efforts. We need the following lemma implicitly established by Chazelle et al. in \cite{Ch95}. The technical proof is postponed to the end of this section.
\begin{lemma}[\bf Chazelle et al. \cite{Ch95}]
\label{TLC}
\it{Let $d\geq 3$ and $P\subset \mathbb{R}^d$ be a set of $n/r$ points such that any $d$ points of $P$ are affinely independent. Assume that we have an affine space $h$ given by $x_1=a_1, \ldots, x_j=a_j$ with $1\leq j \leq d-2$, and a set $\mathcal{S}$ of at least $\alpha_j n^{j+1}/r^{j+1}$ $(j+1)$-tuples of $P$ with $\alpha_{j}>0$ such that the corresponding simplices intersect $h$ exactly once. Then given $n\geq 4(j+1)r/\alpha_j$, there is an $\alpha_{j+1}>0$ and an affine space  $x_1= a_{1}, \ldots , x_j=a_{j}, x_{j+1}=a_{j+1}$ intersecting at least $\alpha_{j+1}n^{j+2}/r^{j+2}$ $(j+1)$-simplices spanned by $(j+2)$-tuples from $P$. Moreover, each such $(j+2)$-tuple has the form $\{p_1,\ldots , p_{j+1}\}\cup\{p_1, \ldots p_{j}, q_{1}\}$ for  $\{p_1, \ldots p_{j+1}\}$, $\{p_1, \ldots p_{j}, q_{1}\}\in \mathcal{S}$. Finally, $a_{j+1}\in [\{p_1, \ldots p_{j+1}\}, \{p_1, \ldots p_{j}, q_{1}\}]$, where by abuse of notation $\{p_1, \ldots p_{j+1}\}$ is the projection of the intersection point of the corresponding $j$-simplex with $h$ onto $x_{j+1}$-axis.}
\end{lemma}
%Namely, the projection of the intersection point of the simplex induced by $\{p_{1}(t), \ldots, p_{j+1}(t)\}$ onto the $x_{j+1}$-axis is simply denoted by $\{p_{1}(t), \ldots, p_{j+1}(t)\}$.

 Assume that we have defined $\lambda_i$, $\gamma_i$ for $i\leq j$, where $1 \leq j\leq d-2$. Let $h$ be a moving space of step $j$ such that at least  $\gamma_j n^{j+1}/r^{j+1}$ $j$-simplices spanned by $C\cap P(t)$ intersect $h(t)$ once in their relative interior. Let us assume that $n\geq 4(j+1)r/\gamma_j$. In what follows, we use the same notation as in the statement of Lemma \ref{TLC}.  By this lemma (used with $\alpha_{j}=\gamma_{j}$, the affine space $h(t)$, and the set of points $C\cap P(t)$), we get a point $a_{j+1}$ contained in at least $\alpha_{j+1} n^{j+2}/r^{j+2}$ intervals  $[\{p_1(t), \ldots p_{j+1}(t)\}, \{p_1(t), \ldots p_{j}(t), q_{1}(t)\}]$ as in the statement of Lemma \ref{TLC}. This is true, because we distinguish two intervals that do not arise from the same pair of $(j+1)$-tuples. We sometimes refer to the projection $\{p_1(t), \ldots, p_{j+1}(t)\}$ as a vertex.

 Set $J=\{x_{\tilde{h},j+1}(t):\tilde{h}  \textnormal{ is a moving space induced by } h \}$. Let $y_1$ be the biggest  $a\in J$ smaller or equal to $a_{j+1}$ (if no such $a$ exists take $-\infty$).
Similarly, let $y_2$ be the smallest  $a\in J$ bigger or equal to $a_{j+1}$(if no such $a$ exists take $\infty$). The following lemma shows that by an appropriate choice of  $\lambda_{j+1}$, not many intervals as above can lie strictly between $y_1$ and $y_2$.
\begin{lemma}
If $\lambda_{j+1}=2\alpha_{j+1}/3(j+1)$ then at most $\alpha_{j+1}n^{j+2}/3r^{j+2}$ intervals as above are strictly contained in $]y_1, y_2[$ on $x_{j+2}$-axis.
\end{lemma}
%We claim that at most $\alpha_{j+1}n^{j+2}/3r^{j+2}$ among intervals as above are contained strictly between $y_1$ and $y_2$ on $x_{j+2}$-axis

%Since, $\textnormal{proj}_{x_{j+1}}(p^{h,\{p_1,\ldots,p_{j+1}\}}(t))=(\textnormal{proj}_{x_{j+1}}\circ p^{h,\{p_1,\ldots,p_{j+1}\}})(t)$ and by definition of $P_h$ we have $\textnormal{proj}_{x_{j+1}}\circ p^{h,\{p_1,\ldots,p_{j+1}\}} \in P_h$,
\begin{proof}
 By contradiction, assume that $\geq \alpha_{j+1}n^{j+2}/3r^{j+2}$ intervals are contained in $]y_1 ,y_2[$. In what follows, we distinguish two vertices arising from different $(j+1)$-tuples. Counted with multiplicities, there are at least $\geq 2\alpha_{j+1}n^{j+2}/3r^{j+2} $  vertices $\{p_{1}(t), \ldots , p_{j+1}(t)\}$ in $]y_1,y_2[$. Each vertex  $\{p_{1}(t), \ldots , p_{j+1}(t)\}$ is counted at most $(j+1) n/r$ times, since there are at most $j+1$ choices of $\{p_{i_1}(t),\ldots, p_{i_j}(t)\}\subset \{p_{1}(t),\ldots ,p_{j+1}(t)\}$ and at most $n/r$ choices for $q(t)$ so that $[\{p_{1}(t),\ldots, p_{j+1}(t)\},$ $ \{p_{i_1}(t), \ldots, p_{i_j}(t), q(t)\}]$ is an interval as above. Hence, there are at least $\geq 2\alpha_{j+1} n^{j+1}/3(j+1)r^{j+1} $ distinct vertices in $]y_1,y_2[$, a contradiction with the value of $\lambda_{j+1}$.

 To see this, we recall that each vertex $\{p_1(t), \ldots, p_{j+1}(t)\}$ is the projection of $p^{h,\{p_1,\ldots,p_{j+1}\}}(t)$ onto $x_{j+1}$-axis for $p^{h,\{p_1,\ldots,p_{j+1}\}} \in F_h$. Since the number of vertices $\{p_1(t), \ldots, p_{j+1}(t)\}$ in $]y_1,y_2[$ is at least $\geq 2\alpha_{j+1} n^{j+1}/3(j+1)r^{j+1}$, the number of $p^{h,\{p_1,\ldots,p_{j+1}\}} \in F_h$ such that the projection of $p^{h,\{p_1,\ldots,p_{j+1}\}}(t)$ onto $x_{j+1}$-axis is in $ ]y_1,y_2[$ is obviously also $\geq 2\alpha_{j+1} n^{j+1}/3(j+1)r^{j+1}$.  Hence, by definition of $P_{h}$ the number of $p\in P_{h}$ such that $p(t)\in]y_1,y_2[ $ is at least $$\frac{2\alpha_{j+1} n^{j+1}}{3(j+1)r^{j+1}}=\frac{\lambda_{j+1} n^{j+1}}{r^{j+1}}> \frac{\lambda_{j+1}\binom{n}{j+1}}{r^{j+1}} = \frac{\lambda_{j+1}|P_{h}|}{r^{j+1}}. $$ Thus, since $N_h$ is a strong $\frac{\lambda_{j+1}}{r^{j+1}}$-net for the kinetic hypergraph of $P_h$ with respect to intervals, there should be a point $w \in N_{h}$ such that $w(t)$ is in $]y_{1},y_{2}[$. This means that there is a moving affine space induced by $h$ whose $x_{j+1}$-coordinate at $t$ $w(t)$ is strictly  between $y_1$ and $y_2$, which contradicts the definition of $y_1$ or $y_2$.
\end{proof}

 Let us set $\lambda_{j+1}=2\alpha_{j+1}/3(j+1)$. By the pigeonhole principle and the lemma above, $y_1$ or $y_2$ belongs to at least $ \alpha_{j+1}n^{j+2}/3 r^{j+2} $ intervals as above (say w.l.o.g. $y_1$). Let us denote by $h_1$ a moving space induced by $h$ such that the $x_{j+1}$-coordinate of $h_1(t)$ is $y_1$. Thus,  at least $ \alpha_{j+1}n^{j+2}/3 r^{j+2} $ $(j+1)$-simplices spanned by $C\cap P(t)$ intersect $h_1(t)$. One needs to be careful, since some of these simplices may intersect $h_1(t)$ more than once or not in their relative interior. However, assuming that $n\geq c_{\alpha_{j}/3} r$ where $c_{\alpha_{j}/3}$ is as in Lemma \ref{Lemma9} one can apply this lemma to conclude that at least $\alpha_{j+1}n^{j+2}/6 r^{j+2} $ of them intersect $h_1(t)$ only once and in their relative interior. Hence, setting $\gamma_{j+1}=\alpha_{j+1}/6$ completes the induction.

 Note that we still need to define $\lambda_d$. Let us set $\lambda_d= \gamma_{d-1}$. It remains us to see that the resulting $N$ is a kinetic weak $\frac{1}{r}$-net for $P$. From the definition of $\gamma_{d-1}=\lambda_d$, we know that some affine space $h(t)$ where $h$ is a moving space of step $d-1$, i.e. a moving line of step $d-1$, must intersect at least $\lambda_d n^d/r^d  > \lambda_d\binom{n}{d}  /r^d = \lambda_d|F_{h}|  /r^d $ $(d-1)$-simplices spanned by $C\cap P(t)$ once in their relative interior. By definition of $F_h$, this implies that there exist $>\lambda_d|F_{h}|  /r^d$ points $p$ of $F_{h}$ such that $p(t)$ belongs to the segment $ C\cap h(t)$. Since $P_{h}$ is the projection of $F_{h}$ onto $x_d$-axis, there exist $>\lambda_d|P_h|/r^d$  points $p'$ of $P_h$ such that $p'(t)$ belongs to the projection of the segment $C\cap h(t)$ onto $x_d$-axis. Hence, since $N_h$ is a strong $\frac{\lambda_{d}}{r^d}$-net for the kinetic hypergraph of $P_h$ with respect to intervals, the projection of $C\cap h(t)$ onto $x_d$-axis must contain a point $v(t)$ of $N_{h}(t)$. By definition of $N$, the moving point $q=(x_{h,1}, \ldots, x_{h,d-1},v)$ belongs to it. Hence, $q(t)\in C$.  Thus, $N$ is a kinetic weak $\frac{1}{r}$-net for $P$, and the theorem follows.

 %Arguing exactly as in case $d=2$, we conclude that there exists a point $w\in N_h$ such that $w(t)$ belongs to the projection of the segment $C\cap h(t)$ onto $x_{d}$-axis. Hence, the point $q=(x_{h,1}, \ldots, x_{h,d-1}, w)$ is in $N$ and $q(t)\in C$. Thus, $N$ is a kinetic weak $\frac{1}{r}$-net for $P$, and the theorem follows.

\end{proof}
\begin{lemma}
\label{TP}
\it{Let $P$ be a set of points moving polynomially in $\mathbb{R}^d$ with bounded description complexity $\beta$. Let $\{p_1,\ldots, p_{j+1}\}$ be a $(j+1)$-tuple from $P$ and $h$ some moving affine space of step $j$, as defined in the proof of Theorem \ref{MT}. Then one can define a moving point $p$ such that for each $t\geq 0$ when the intersection of \textnormal{aff($\{p_{1}(t),\ldots, p_{j+1}(t)\}$)} and $h(t)$ is a single point, it is equal to $p(t)$. Moreover,  $p$ has description complexity $f(j+1)$, where $f: \{1,\ldots, d\} \to \mathbb{N}$ is some increasing function with $f(1)=\beta$.}
\end{lemma}
%[Proof of the Lemma \ref{TP}]

\begin{proof}
The case where for each $t\geq 0$ the intersection of \textnormal{aff}($\{p_{1}(t),\ldots, p_{j+1}(t)\}$) and $h(t)$ is empty or contains more than one point is trivial, since one can define $p$ to be static.

 Hence, one can assume that for some $t\geq 0$ the intersection above contains a single point. We prove the lemma by induction on the step. Observe that the function defining the first coordinate of a moving space of step $i$ is obtained by projection of some point from $P$, hence has description complexity  $\beta=f(1)$.

 Asssume that the lemma holds for moving points arising from moving spaces of step at most $0\leq j-1\leq d-2$. Let $p_1, \ldots, p_{j+1}$ be any $(j+1)$-tuple of points from $P$ and $h$ any moving space of step $j$ and given by $x_1 = x_{h,1}, \ldots, x_{j}= x_{h,j}$. Then it follows from the definition of $x_{h,i}$ (see Theorem \ref{MT}), the induction hypothesis, and the observation above that $x_{h,i}$ has description complexity $f(i)$. Assume $h(t)$ and aff$(\{p_{1}(t), \ldots , p_{j+1}(t)\})$ intersect in a unique point $p(t)$. Then we can write  $p(t)=\alpha_{1}(t)p_{1}(t)+\ldots + \alpha_{j+1}(t)p_{j+1}(t)$ and from the general position assumption the points $p_1(t),\ldots, p_{j+1}(t)$ are affinely independent, so a point of aff$(\{p_1(t),\ldots p_{j+1}(t)\})$ is uniquely determined by an affine combination of the points $p_i(t)$. An immediate consequence from the unicity of $\alpha_i(t)$ is the following matricial equivalence:

\begin{center}
\[  \left( \begin{array}{ccc}
[p_1(t)]_1 & \ldots & [p_{j+1}(t)]_1   \\

\vdots &   & \vdots   \\

[p_1(t)]_{j} &  \ldots & [p_{j+1}(t)]_{j}   \\

1 & \ldots & 1   \\

 \end{array} \right)
 \left( \begin{array}{c}
\alpha_1(t)  \\

\vdots    \\

\vdots \\

\alpha_{j+1}(t)

 \end{array} \right)
=
 \left( \begin{array}{c}
x_{h,1}(t)  \\

\vdots    \\

x_{h,j}(t)  \\

1

 \end{array} \right)\]

$$\Longleftrightarrow$$
 \[ \left( \begin{array}{c}
\alpha_1(t)  \\

\vdots    \\

\vdots \\

\alpha_{j+1}(t)

 \end{array} \right)
=
 \left( \begin{array}{ccc}
[p_1(t)]_1 & \ldots & [p_{j+1}(t)]_1   \\

\vdots &   & \vdots   \\

[p_1(t)]_{j} &  \ldots & [p_{j+1}(t)]_{j}   \\

1 & \ldots & 1   \\

 \end{array} \right)^{-1}
 \left( \begin{array}{c}
x_{h,1}(t)  \\

\vdots    \\

x_{h,j}(t)  \\

1

 \end{array} \right)
 \]

\end{center}

%et us see why the inverse of the matrix above exists. By assumption, the intersection point of aff$(\{p_1(t),\ldots p_{j+1}(t)\})$ with $h(t)$ is unique. Hence, the solution to the topmost matricial equation is unique and the matrix above is invertible.

It follows from the Cramer's rule that the moving point $\alpha_i$, whose position at any $t\geq 0$ is $\alpha_i(t)$ given by the equation above has  description complexity depending only on $j$ and $f(j)$. Hence, the moving point $p$ whose position at $t$ is $\alpha_{1}(t)p_{1}(t)+\ldots + \alpha_{j+1}(t)p_{j+1}(t)$ also has  description complexity depending only on $j$ and $f(j)$ that we denote by $f(j+1)$ (w.l.o.g. $f(j+1)\geq f(j)$). This completes the proof.
\end{proof}

 %Moreover, assume that no $d+2$ points are contained in the same hyperplane.
\begin{lemma}
\label{Lemma9}
\it{Let $1\leq j \leq d-1$ and $P\subset \mathbb{R}^d$ be a set of $n/r$ points such that no $d+2$ of them lie in a hyperplane. Assume that we have a set $\mathcal{S}$ of $\alpha n^{j+1}/r^{j+1}$ $(j+1)$-tuples from $P$ such that the convex hull of each of them intersects a given affine space $V$ of dimension $d-j$. Then there exists $c_{\alpha}$  such that if $n\geq c_{\alpha}r$, then there are at least $\alpha n^{j+1}/2r^{j+1}$ $(j+1)$-tuples from $\mathcal{S}$ such that their convex hulls intersect $V$ exactly once and in their relative interior.}
\end{lemma}
\begin{proof}
 We can assume that $\alpha > 0$, otherwise there is nothing to show. Assume that the convex hulls of at least $\alpha n^{j+1}/2r^{j+1}$  $(j+1)$-tuples from $\mathcal{S}$ intersect the affine space $V$ more than once or on their relative boundary. We will show that for $n\geq c_{\alpha} r$, where $c_{\alpha}$ is large enough, we obtain a contradiction. When the convex hull of a $(j+1)$-tuple $A$ intersects $V$ more than once, one can take two intersection points $x_1$ and $x_2$ with the affine space $V$ and follow the line passing through $x_1$, $x_2$ until the relative boundary of conv($A$) is intersected. Hence, in both cases the relative boundary of conv($A$) must be intersected. Clearly, this means that there is a subset of $j$ points  from $A$ whose convex hull intersects $V$. Each such $j$-tuple can be counted at most $n/r$ times. Hence, there are at least $\alpha n^{j}/2  r^{j}$ distinct $j$-tuples arising from elements of $\mathcal{S}$ as above.

We define $\mathcal{S}_j$ to be the set of $j$-tuples above, i.e., those whose convex hulls intersect $V$. In order to obtain a contradiction, we consider the following iterative procedure. Set $\gamma_j= \alpha/2$.  Assume that $\mathcal{S}_i$ was defined for some $2\leq i \leq j$ and contains at least  $\gamma_{i} n^{i}/r^{i}$  $i$-tuples whose convex hulls intersect $V$. We say that $\mathcal{S}_i$ is \emph{good} if it has a subset of at least $\gamma_{i} n^{i}/2r^{i}$ $i$-tuples, denoted by $\mathcal{G}_i$, such that the convex hull of no $(i-1)$-tuples which are $(i-1)$-subsets of the $i$-tuples from $\mathcal{G}_i$ intersects the affine space $V$. Otherwise, we say that the set $\mathcal{S}_i$ is \emph{bad}, and define $\mathcal{S}_{i-1}$ to be the set of $(i-1)$-tuples whose convex hulls intersect $V$ and  each of them is contained in some $i$-tuple from $\mathcal{S}_i$. Clearly, the size of $S_{i-1}$ is at least $\gamma_{i} n^{i-1}/2r^{i-1}$, since an $(i-1)$-tuple can appear in at most $n/r$ $i$-tuples of $S_{i}$. Finally, we set $\gamma_{i-1}=\gamma_{i}/2$. For some $i$  the procedure must stop with a good $\mathcal{S}_i$. Indeed, if we had to compute $\mathcal{S}_1$, then this means that we have a set of points from $P$ of cardinality at least $\gamma_{1} n/r$ such that each point belongs to $V$. This means that for  $n$ large enough ($n \geq (d+2)r/\gamma_{1}$), we get a set of at least $d+2$ points contained in $V$. That is, an affine space of dimension at most $d-1$, a contradiction.

Hence, we can assume that $S_i$ is good for some $i$. Let $\mathcal{G}_i$ be as above. Define a graph $G$ whose vertices are the different $(i-1)$-tuples each contained in some $i$-tuple from $\mathcal{G}_i$. For each $i$-tuple from $\mathcal{G}_i$ choose two different $(i-1)$ subsets and connect them by an edge. The number of edges is at least  $\gamma_{i} n^{i}/2 r^{i}$, since an edge determines the $i$-tuple it arises from. %By the pruning lemma,
 Clearly, there is  a vertex of degree at least $\gamma_{i} n^{i}/2 r^{i}\binom{n/r}{i-1}\geq \gamma_{i} n/2r$.  Take one such $(i-1)$-tuple $\{p_1, \ldots, p_{i-1} \}$. This means that the  affine space given by  aff$(V,p_1, \ldots , p_{i-1})$ of dimension at most $d-1$ contains at least $i-1+ \gamma_{i} n/2r$ points, i.e., $p_1, \ldots, p_{i-1}$ and the points of the union of all neighbours of $\{p_1, \ldots, p_{i-1}\}$ in $G$.
Indeed, let $p$ be the intersection point of conv($p_1, \ldots , p_i)$ with $V$, where $p_i$ belongs to some neighbour of $\{p_1, \ldots, p_{i-1} \}$ in $G$. We show that aff($\{p_1, \ldots , p_{i-1},p\}$)=aff($\{p_1, \ldots , p_{i-1},p_{i}\}$).  If $p_i$ is in aff($\{p_1, \ldots , p_{i-1}\}$) the equality is clear. If not, then aff($\{p_1, \ldots , p_{i-1}, p\}$) has dimension strictly bigger than aff($\{p_1, \ldots , p_{i-1}\}$) while being contained in aff($\{p_1, \ldots , p_{i-1}, p_{i}\}$), so the equality holds. Hence, for $n$ large enough ($n \geq  (d+1)2r/\gamma_i)$ we get a contradiction, since strictly more than $d+2$ points are in the  affine space aff$(\{V,p_1, \ldots , p_{i-1}\})$ whose dimension is at most $d-1$, in particular, the points are contained in a hyperplane.

  %Indeed, any simplex given by the convex hull of  $p_{1},\ldots , p_{i}$, where the tuple $p_{1},\ldots , p_{i}$ is in $G_i$, is contained in aff$(p_{1},\ldots, p_{j-1}, p)$ where $p$ denotes the intersection point of conv($p_1, \ldots , p_i)$ with $V$, because $p$ is not in conv$(p_1, \ldots , p_{i-1})$, but is in conv$(p_1, \ldots , p_{i})$.

 %Progressively delete each vertex of degree $\leq \teta_{i}/2 n^{i}/4  r^{i}\binom{n}{i-1}$. Clearly, this procedure cannot delete more than half of the edges, in particular we obtain a
%we get a graph such that each remaining vertex has degree at least $\alpha_j n/8 r$. Take one such $j-1$-tuple $\{p_1, \ldots, p_{j-1} \}$. This means that the  affine space given by  aff$(V,p_1, \ldots , p_{j-1})$ of dimension at most $d-1$ contains at least $j+ \alpha n/8r$ points. Indeed, any simplex given by the convex hull of  $p_{1},\ldots , p_{j}$ is contained in aff$(p_{1},\ldots, p_{j-1}, i)$ where $i$ denotes the intersection point of conv($p_1, \ldots , p_j)$ with $V$, because $i$ is not in conv$(p_1, \ldots , p_{j-1})$, but is in conv$(p_1, \ldots , p_{j})$. For an $n$ big enough, we get a contradiction since strictly more than $d+1$ points are in the same affine space whose dimension is at most $d-1$, in particular the points are contained in the same hyperplane.

\end{proof}

\begin{proof}[Proof of Lemma \ref{TLC}]
Define the hypergraph on $P$ whose hyperedges are the different $(j+1)$-tuples of $\mathcal{S}$. Iteratively remove a $j$-tuple $A$ from $\binom{n/r}{j}$ and remove the $(j+1)$-tuples containing it from $\mathcal{S}$  if the number of the remaining elements from $\mathcal{S}$ containing $A$ is at most $\alpha_j n^{j+1}/ 2r^{j+1}\binom{n/r}{j}$. Call $\mathcal{S}'$ the remaining set of $(j+1)$-tuples. This procedure cannot remove more than $\alpha_j n^{j+1}/2r^{j+1}$ hyperedges, so the resulting hypergraph is not empty and each $j$-tuple contained in some element from $\mathcal{S}'$ is contained in  $$> \frac{\alpha_j n^{j+1}}{2 r^{j+1} \binom{n/r}{j}}\geq \frac{\alpha_j n}{2r}= \frac{\alpha' n}{r}$$ elements from $\mathcal{S}'$, where we set $\alpha'=\alpha_j/2 $.

%The number of such intervals is at least
We now project the intersections of simplices corresponding to $(j+1)$-tuples from $\mathcal{S}'$ with $h$ onto the $x_{j+1}$-axis. For the sake of simplicity, the projection of the intersection point induced by the tuple $\{p_1, \ldots, p_{j+1}\}$ will still be denoted by $\{p_1, \ldots, p_{j+1}\}$.
Connect $\{p_1, \ldots, p_{j+1}\}$ with $\{q_1, \ldots , q_{j+1} \}$ if there is a sequence  $\{p_1, \ldots, p_{j+1}\}$, $\{p_1, \ldots, p_{j},q_{1} \}, \ldots,$ $ \{q_1, \ldots, q_{j+1}\}$, where each member of the sequence is an element of $\mathcal{S}'$ and the points $p_1, \ldots, p_{j+1}, q_{1}, \ldots , q_{j+1}$ are all distinct. We sometimes refer to such an interval as \emph{type 1} interval. The following procedure gives a lower bound on the number of such intervals (we distinguish two intervals arising from different pairs of $(j+1)$-tuples): Choose any $\{p_1, \ldots, p_{j+1}\}$ in $\mathcal{S}'$. Take any $q_1$ such that  $\{p_1,\ldots, p_j, q_1\}$ is in $\mathcal{S}'$ with $q_1$ different from $p_{j+1}$. Then take any $q_2$ such that $q_2$ is different from $p_j,~p_{j+1}$ and $\{p_1,\ldots, p_{j-1}, q_1, q_2\}$ is in $\mathcal{S}'$ etc. The lower bound below follows  $$\frac{|\mathcal{S}'|(\alpha' n/r -j-1)^{j+1}}{2(j+1)!}\geq\frac{|\mathcal{S}'|(\alpha' n/2r)^{j+1}}{2(j+1)!} $$ given $\alpha' n/2r\geq j+1$. Indeed, starting from $\{p_1, \ldots, p_{j+1}\}$ an interval $[\{p_1, \ldots, p_{j+1}\},$ $\{q_1, \ldots , q_{j+1} \}]$ is counted at most once  for each permutation of $q_1, \ldots , q_{j+1}$. Thus from the one dimensional selection lemma, see \cite{Ar91}, we know that there exists a point $a_{j+1}$ contained in at least $$\frac{|\mathcal{S}'|^2[(\alpha' n/2r)^{j+1}/2(j+1)!]^2}{4|\mathcal{S}'|^2}=\frac{1}{4}\frac{[(\alpha' n/2r)^{j+1}]^2}{[2(j+1)!]^2}= \frac{\alpha'' n^{2j+2}}{r^{2j+2}} $$
intervals, where we set  $\alpha''=\alpha'^{2j+2}/2^{2j+6}[(j+1)!]^2$.

 Clearly, if a point is contained in an interval $[\{p_1, \ldots, p_{j+1}\},$ $ \{q_1, \ldots , q_{j+1} \}]$, it must also be contained in some interval $[\{p_1, \ldots p_s, q_{1}, \ldots q_{j-s+1}\}, \{p_1, \ldots p_{s-1}, q_{1}, \ldots q_{j-s+2}\}]$. These latter kind of intervals are refered to as \emph{type 2} intervals.  Moreover, an interval of type 2 can be counted at most $(j+1)(jn/r)^j$ times. Indeed, there are at most $j+1$ possible positions for such an interval in a chain as above (used to define type 1 intervals), at most $j$ possibilities of choosing a point that is replaced in a $(j+1)$-tuple while a subchain is extended, and at most $n/r$ candidates to replace such a point. Hence, $a_{j+1}$ is contained in at least $\alpha'' n^{2j+2}/r^{2j+2}(j+1)(jn/r)^j=\alpha '''n^{j+2}/ r^{j+2}$ intervals of type 2, where $\alpha'''= \alpha''/(j+1)j^j$.

  Each interval of this latter type containing the point $a_{j+1}$ corresponds to a $(j+1)$-simplex spanned by $P$ intersecting the affine space given by $x_1=a_1, \ldots, x_{j+1}=a_{j+1}$. Finally, it is easy to see that a spanned $(j+1)$-simplex arises from at most $(j+2)(j+1)$ intervals of type 2. Hence, there exist at least $\alpha'''n^{j+2}/(j+2)(j+1)r^{j+2}$ $(j+1)$-simplices arising from intervals of type 2 pierced by $a_{j+1}$. Thus, selecting one such interval for each such simplex and defining $\alpha_j=\alpha'''/(j+2)(j+1)$ we can conclude.

\end{proof}

\bibliography{abc}

\end{document}